\documentclass[a4paper,11pt]{article}
\usepackage[left=2.8cm,right=2.8cm,top=2.8cm,bottom=2.8cm]{geometry}
\usepackage{amsmath,amsthm,amsfonts,amssymb,url}
\usepackage[english]{babel}
\usepackage{graphicx}
\usepackage{subfig}

\newtheorem{theorem}{Theorem}

\title{A tree metric using structure and length to capture distinct phylogenetic signals}
\author{Michelle Kendall, Caroline Colijn}
\date{}

\begin{document}
\maketitle

\begin{abstract}
Phylogenetic trees are a central tool in understanding evolution. They are typically inferred from sequence data, and capture evolutionary relationships through time.  It is essential to be able to compare trees from different data sources (e.g.\ several genes from the same organisms) and different inference methods. 
We propose a new metric for robust, quantitative comparison of rooted, labeled trees. It enables clear visualizations of tree space, gives meaningful comparisons between trees, and can detect distinct islands of tree topologies in posterior distributions of trees. 
This makes it possible to select well-supported summary trees.
We demonstrate our approach on Dengue fever phylogenies.
\end{abstract}

\section{Introduction}

Phylogenetic trees are fundamental tools for understanding evolution.
Improvements in sequencing technology have meant that phylogenetic analyses are growing in size and scope. 
However, when a tree is inferred from data there are multiple sources of uncertainty. Competing approaches to tree estimation can produce markedly different trees.
Trees may conflict due to signals from selection (e.g. convergent evolution), and/or when derived from different data (e.g.\ the organisms' mitochondrial vs nuclear DNA, individual genes or other subsets of sequence data~\cite{Jiang2014}). 
Evolution is not always tree-like: species trees differ from gene trees, and many organisms exchange genes through horizontal gene transfer. 
It is therefore crucial to be able to compare trees to identify these signals.

Trees can be compared by direct visualization, aided by methods such as tanglegrams and software such as DensiTree~\cite{DensiTree2010}, but this does not lend itself to detailed comparison of large groups of trees. 
Current quantitative methods for tree comparison suffer from the challenges of visualizing non-Euclidean distances~\cite{Harris2008} and from counter-intuitive behavior.
For example, the nearest-neighbor interchange (NNI) distance  of Robinson and Foulds (RF)~\cite{Robinson1981}, which is the most widely used, is hampered by the fact that large NNI distances do not imply large changes among the shared ancestry of most tips~\cite{Steel1993,Lin2011,Kuhner2014}.
In fact, two trees differing in the placement of a single tip can be a maximal NNI distance apart. 

We introduce a \emph{metric} which flexibly captures both tree structure and branch lengths. It can be used as a quantitative tool for comparing phylogenetic trees.    
Each metric on trees defines a \emph{tree space}; this tree space lends itself to clear visualizations in low dimensions, and captures and highlights differences in trees according to their biological significance.

In Section~\ref{sec:metric} we formally define our distance function, prove that it is a metric, and explain its capacity to capture tree structure and branch lengths.
We also provide a brief survey, explaining how our metric relates to and differs from existing metrics (Section~\ref{sec:others}).
In Section~\ref{sec:exploring} we explain some of the applications of our metric. We show how our metric enables visualization of tree space (Section~\ref{sec:visualisation}) and detection of islands (Section~\ref{sec:islands}), which we demonstrate with a simple application to Dengue fever phylogenies. We also explain how our metric provides a new suite of methods for selecting summary trees in Section~\ref{sec:summary}.
We conclude with some ideas for extensions to our metric in Section~\ref{sec:conclusion}.

\section{Metrics} \label{sec:metric}

\subsection{Our metric: definition and proof} \label{sec:defn}
Let $\mathcal{T}_k$ be the set of all rooted trees on $k$ tips with labels ${1,\dots,k}$.
In common with previous literature~\cite{Harding1971,Robinson1981} we say that trees $T_a, T_b \in \mathcal{T}_k$ have the same labeled shape or \emph{topology} if the set of all tip partitions admitted by internal edges of $T_a$ is identical to that of $T_b$, and we write this as $T_a \cong T_b$. 
We say that $T_a=T_b$ if they have the same topology and each corresponding branch has the same length.

For any tree $T_a\in \mathcal{T}_k$ let $m_{i,j}$ be the number of edges on the path from the root to the most recent common ancestor (MRCA) of tips $i$ and $j$, 
let $M_{i,j}$ be the length of this path, and 
let $p_i$ be the length of the pendant edge to tip $i$. 
Then, including all pairs of tips, we have two vectors:
$$ m(T) = (m_{1,2}, m_{1,3},\dots,m_{k-1,k},\underbrace{1,\dots,1}_{k \text{ times}}) \enspace ,$$
which captures the tree topology, and 
$$ M(T) = (M_{1,2}, M_{1,3},\dots,M_{k-1,k},p_1,\dots,p_k) $$
which captures the topology and the branch lengths. 
The vector $M(T)$ is similar to the vector of cophenetic values~\cite{Sokal1962,Cardona2013} (Section~\ref{sec:others}).
We form a convex combination of these vectors, parameterized with $\lambda \in [0,1]$, to give
$$
v_{\lambda}(T) = (1-\lambda)m(T) + \lambda M(T) \enspace .
$$

\noindent Figure~\ref{fig:construction} provides an example of this calculation for two small trees.

\begin{figure}[htb]
	\centering	
    \includegraphics[width=0.95\textwidth]{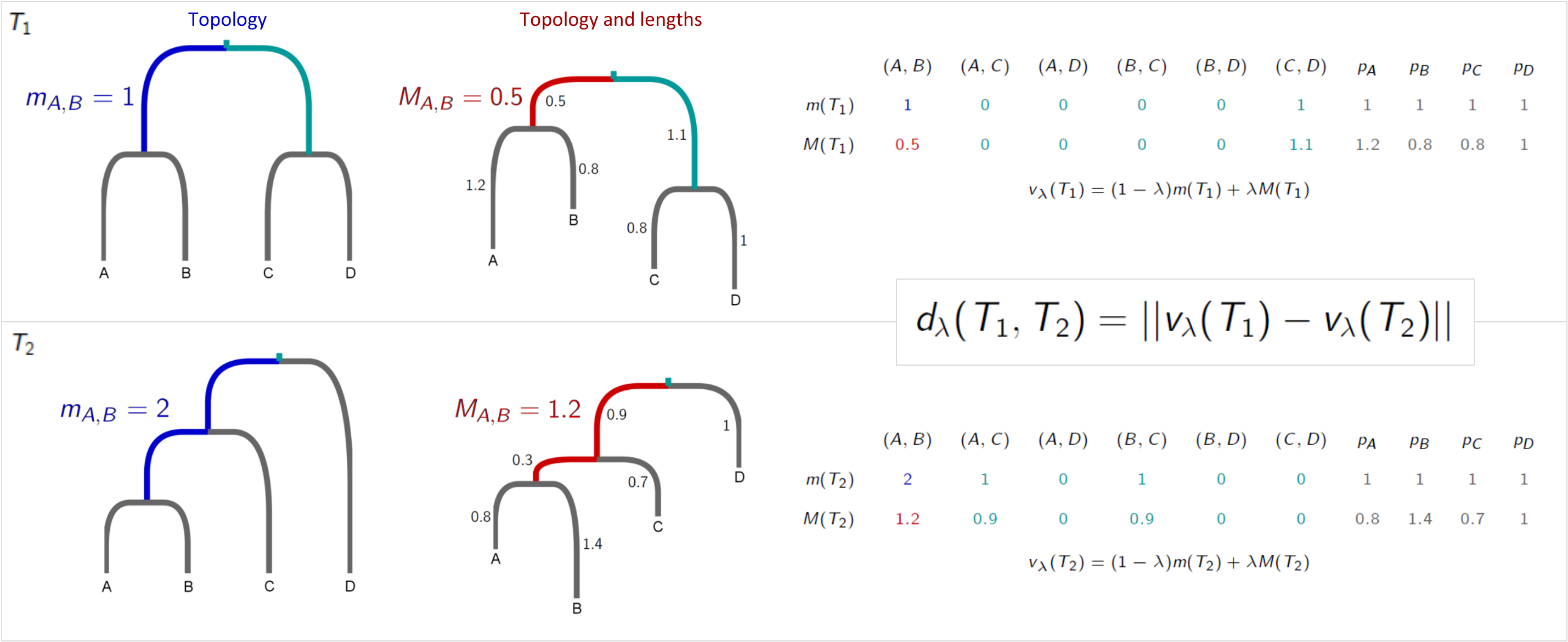}
\caption{A tree is characterized by the vectors $m$ and $M$, which are calculated as shown. These are used to calculate the distance between the trees for any $\lambda \in [0,1]$. Here, $d_0(T_1,T_2)=2$ and $d_1(T_1,T_2)=1.96$.}
\label{fig:construction}
\end{figure}

A metric is a mathematical notion of distance; specifying a metric gives structure and shape to a set of objects, forming a \emph{space}.
A function $d(T_1,T_2)$ is a metric if, for all $T_1, T_2 \in \mathcal{T}_k$, 
\begin{enumerate}
\item $d(T_1,T_2)\geq 0$ (distances are non-negative)
\item $d(T_1,T_2)=0 \Leftrightarrow T_1=T_2$ (the distance is only 0 if they are the same) 
\item $d(T_1,T_2)=d(T_2,T_1)$ (distance is  symmetric)
\item for any $T_3 \in \mathcal{T}_k$, $d(T_1,T_2) \leq d(T_1,T_3)+d(T_3,T_2)$ (the triangle inequality) 
\end{enumerate}

\begin{theorem}
The function $d_{\lambda}:\mathcal{T}_k \times \mathcal{T}_k \rightarrow \mathbb{R}$ given by 
$$ d_{\lambda}(T_a, T_b) = \| v_{\lambda}(T_a) - v_{\lambda}(T_b)\| $$
is a metric on $\mathcal{T}_k$, where $\| \cdot \|$ is the Euclidean distance ($l^2$-norm) and $\lambda \in [0,1]$.
\end{theorem}

\begin{proof}
Since the Euclidean distance between vectors satisfies the conditions (1), (3) and (4) for being a metric, it remains to prove that $d_0(T_a,T_b)=0 \Leftrightarrow T_a\cong T_b$ 
(i.e.\ the distance is 0 with $\lambda=0$ if and only if the trees have the same topology) 
and $d_\lambda(T_a,T_b)=0 \Leftrightarrow T_a=T_b$ for all $\lambda \in (0,1]$ 
(i.e.\ the distance is 0 for $0 < \lambda \leqslant 1$ if and only if the trees are 
identical).
We will address this in three stages, showing that (1) the tree topology vector, (2) the branch-length focused vector, and (3) their convex combination each uniquely define a tree. 
That is, we show that for $T_a, T_b \in \mathcal{T}_k$,
\begin{enumerate}
\item $m(T_a)=m(T_b) \Leftrightarrow T_a \cong T_b$,
\item $M(T_a)=M(T_b)\Leftrightarrow T_a = T_b$, and 
\item for $\lambda \in (0,1), v_\lambda(T_a)=v_\lambda(T_a) \Leftrightarrow T_a = T_b$.
\end{enumerate}
For ease of notation we restrict our attention here to binary trees; it is straightforward to extend these arguments to trees that are not binary. 

\paragraph{1.} We show that $m(T)$ characterizes a tree topology. 
Suppose that for $T_a, T_b \in \mathcal{T}_k$ we have $d_0(T_a,T_b)=0$, so $m_{i,j}(a)=m_{i,j}(b)$ for all pairs $i,j \in {1,\dots,k}$. 
Consider the tip partition created by the root of $T_a$.
That is, if the root and its two descendant edges were removed, then $T_a$ would be split into two subtrees, whose tip sets we label $L$ and $R$.
For all leaf pairs $(i,j)$ with $i \in L$ and $j \in R$ we have $m_{i,j}(a) = 0$, and therefore $m_{i,j}(b)=0$.
Thus the root of $T_b$ also admits the partition $\{L,R\}$.

Similarly, any internal node $n$ in $T_a$ partitions its descendant tips into non-empty sets $L_n$, $R_n$.
Let the number of edges on the path from the root to $n$ be $x_n$.
For all leaf pairs $(i,j)$ with $i \in L_n$, $j \in R_n$ we have $m_{i,j}(a) = x_n =m_{i,j}(b)$, and so there must also be an internal node in $T_b$ which partitions the leaves into the sets $L_n, R_n$.
Since this is true for all internal nodes, and hence all internal edges, we have $T_a \cong T_b$, and $d_0$ is a metric on tree topologies.
Note that the final $k$ fixed entries of $m(T)$ are redundant for unique characterization of the topology of the tree, but are included to allow the convex combination of the topological and branch-length focused vectors.

\paragraph{2.} We show that $M(T)$ characterizes a tree using a similar argument to that of part (1). 
Suppose that for $T_a, T_b \in \mathcal{T}_k$ we have $d_1(T_a,T_b)=0$, so $M_{i,j}(a)=M_{i,j}(b)$ for all pairs $i,j \in {1,\dots,k}$. 
Let the \emph{length} of the path from the root to internal node $n$ be $X_n$. 
Then for all $i \in L_n$, $j \in R_n$ we have $M_{i,j}(T_a)=X_n=M_{i,j}(T_b)$, which means that $T_b$ also contains an internal node at distance $X_n$ from the root which admits the partition $\{L_n,R_n\}$.
Since this holds for all internal nodes including the root (where $X_n=0$), we have that $T_a$ and $T_b$ have the same topology and \emph{internal} branch lengths.

The final $k$ elements of $M(T)$ correspond to the pendant branch lengths.
When $M(T_a)=M(T_b)$ we have that for each $i \in {1,\dots,k}$ the pendant branch length to tip $i$ has length $p_i$ in both $T_a$ and $T_b$. 
Thus $T_a$ and $T_b$ have the same topology and branch lengths, hence $T_a=T_b$ and $d_1$ is a metric. 

\begin{figure}[htb]
\subfloat{
\includegraphics[width=0.45\textwidth]{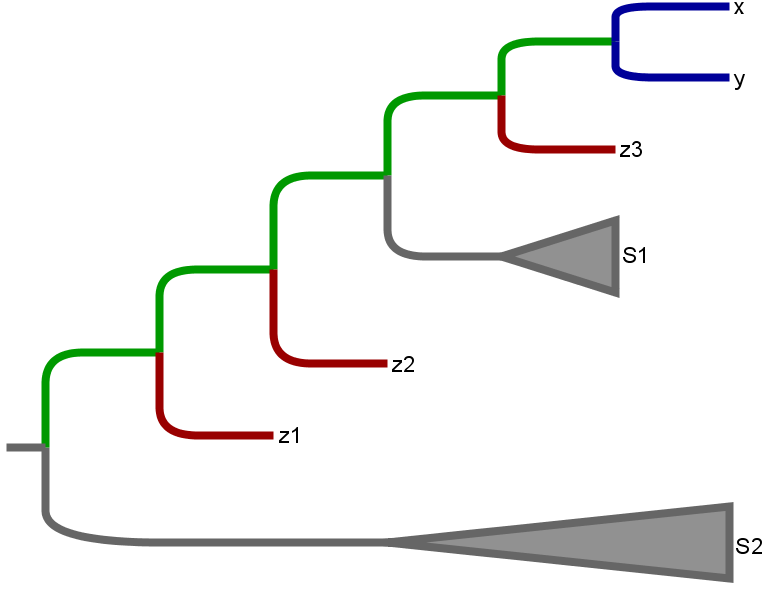}}
\subfloat{
\includegraphics[width=0.45\textwidth]{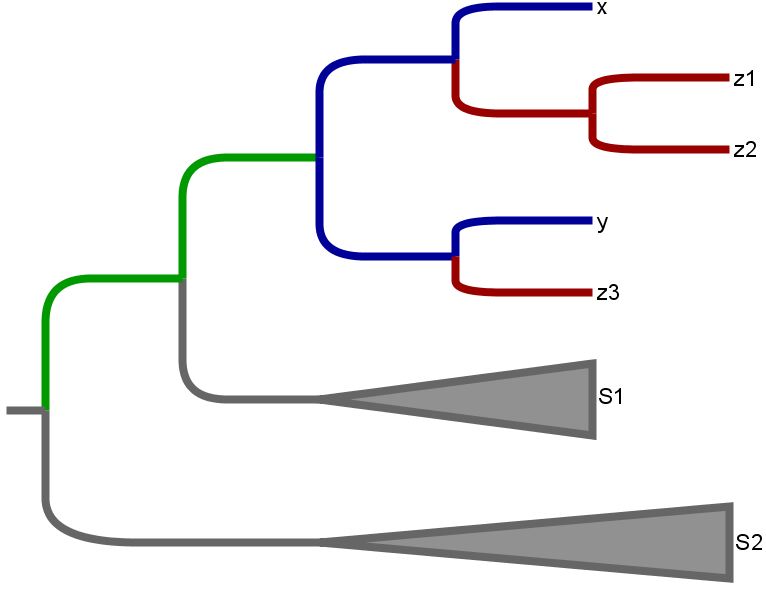}}
\caption{If $d(T_a,T_b)=0$ then $T_a$ and $T_b$ must share the same root partition, hence $S_2$ is the same set of tips in both trees. 
If $m_{x,y}(T_a)\neq m_{x,y}(T_b), m_{x,y}(T_a) - m_{x,y}(T_b) = n$ (here $m_{x,y}(T_a) - m_{x,y}(T_b)=5-2=3$), then there exist at least 3 tips $z_1,z_2,z_3$ between the root and the MRCA of $x$ and $y$ in $T_a$, but positioned further from the root than the MRCA of $x$ and $y$ in $T_b$.}
\label{fig:proof_diagram}
\end{figure}

\paragraph{3.} Finally, we need to show that $v_\lambda(T)$ characterizes a tree for $\lambda \in (0,1)$.
Suppose that for $T_a, T_b \in \mathcal{T}_k$ and $\lambda \in (0,1)$ we have $d_\lambda(T_a,T_b)=0$, so $v_\lambda(T_a)=v_\lambda(T_b)$.

Each vector has length $\binom{k}{2}+k=\frac{k(k+1)}{2}$. 
It is clear that for the final $k$ entries, that is for $\frac{k(k-1)}{2} < i \leq \frac{k(k+1)}{2}$ we have
$$
0 = (1-\lambda)(1-1) + \lambda(M_i(T_a)-M_i(T_b))
$$
which implies that $M_i(T_a)=M_i(T_b)$.

We therefore restrict our attention to the first $\binom{k}{2}$ elements of $v_\lambda$.
Now $d_\lambda(T_a,T_b)=0$ implies that  
\begin{equation}
0=(1-\lambda)(m_{i,j}(T_a)-m_{i,j}(T_b)) + \lambda(M_{i,j}(T_a)-M_{i,j}(T_b)) 
\label{eqn:condition}
\end{equation}
for all $i,j \in {1,\dots,k}$.
We show that, for any $\lambda \in (0,1)$, although it is possible for Equation~\ref{eqn:condition} to hold for \emph{some} $i,j \in {1,\dots,k}$ it will only hold for all $i,j \in {1,\dots,k}$ when $T_a=T_b$.

Suppose for a contradiction that we have $T_a \neq T_b$ but $d_\lambda(T_a,T_b)=0.$
First, observe that if $m_{i,j}(T_a)=0$ then $M_{i,j}(T_a)=0$, which forces $m_{i,j}(T_b)=M_{i,j}(T_b)=0$, and so $d_\lambda(T_a,T_b)=0$ implies that $T_a$ and $T_b$ must share the same root partition.
Now fix $\lambda \in (0,1)$ and consider a pair of tips $x,y \in {1,\dots,k}$ with $m_{x,y}(T_a) \neq m_{x,y}(T_b)$, $m_{x,y}(T_a), m_{x,y}(T_b)\neq 0$, which must exist since $T_a \neq T_b$, using part (1).
Without loss of generality, suppose that $m_{x,y}(a)-m_{x,y}(b)=n$, where $n \in \mathbb{N}$.
Then there exist at least $n$ tips $z_1,\dots,z_n$ for which, because the trees have the same root partition, we have 
$$
m_{x,z_i}(T_a) = m_{y,z_i}(T_a) < m_{x,y}(T_a)
$$ 
and 
$$
m_{x,z_i}(T_b) \geq m_{x,y}(T_b), \quad m_{y,z_i}(T_b) \geq m_{x,y}(T_b) \enspace ,
$$
for each $i \in {1,\dots,n}$ (see Figure~\ref{fig:proof_diagram}). Pick $z_j$ so that $m_{x,z_j}(T_a)=\min_{i \in [n]}{m_{x,z_i}(T_a)}$.
Then $m_{x,z_j}(T_a)-m_{x,z_j}(T_b)\leq m_{x,y}(T_a)-n-m_{x,y}(T_b) = n - n = 0$.
Now since Equation~\ref{eqn:condition} holds for all $i,j \in {1,\dots,k}$, we have
\begin{eqnarray*}
0 \geq m_{x,z_j}(T_a)-m_{x,z_j}(T_b) &=& \left( \frac{\lambda}{1-\lambda} \right) (M_{x,z_j}(T_b) - M_{x,z_j}(T_a)) \\
&\geq& \left( \frac{\lambda}{1-\lambda} \right) (M_{x,y}(T_b) - M_{x,z_j}(T_a)) \\
&=& \left( \frac{\lambda}{1-\lambda} \right) (M_{x,y}(T_b) - M_{x,y}(T_a) + M_{x,y}(T_a) - M_{x,z_j}(T_a))
\end{eqnarray*}
But $M_{x,y}(T_b)-M_{x,y}(T_a)=\left(\frac{1-\lambda}{\lambda}\right)n >0$ and $M_{x,y}(T_a)-M_{x,z_j}(T_a) > 0$ so we have a contradiction.
Thus Equation~\ref{eqn:condition} cannot hold for all $i,j \in {1,\dots,k}$, so $d_\lambda(T_a,T_b)=0 \Rightarrow T_a=T_b$.
\end{proof}

Our metric is fundamentally for \emph{rooted} trees. A single unrooted tree, when rooted in two different places, produces two distinct rooted trees, and our distance between these will be positive. It will be large if the two distinct places chosen for the roots are separated by a long path in the original unrooted tree. However, it would be straightforward to check if two trees have the same (unrooted) topology in our metric: root both trees on the edge to the same tip and find the distance. Re-rooting a tree will induce systematic changes in $v(T)$, with some entries increasing and others decreasing by the same amount.  
The metric $d_\lambda$ is invariant under permutation of labels. 
That is, for trees $T_a$ and $T_b$ and a label permutation $\sigma$, $d_\lambda (T_a,T_b)=d_\lambda (T_a^\sigma,T_b^\sigma)$.

We note that alternative, similar definitions for a metric on $\mathcal{T}_k$ are possible.
In particular, the metric defined by 
$$ 
D_{\lambda}(T_a, T_b) =(1-\lambda) \| m(T_a) - m(T_b)\| + \lambda \| M(T_a) - M(T_b)\|
$$
gives similar behavior to the metric we have used. The difference between the two is that in $D$, the Euclidean distances are taken between the $m$ and $M$ vectors \emph{before} they are weighted by $\lambda$. Rather than a Euclidean distance between two vectors ($v$ for each tree), $D$ is a weighted sum of two different metrics: the distance between $m(T_a)$ and $m(T_b)$ (first term in the above), and between $M(T_a)$ and $M(T_b)$ (second term). 
A benefit of $D_\lambda$ is that it is linear in $\lambda$, so that the changes as $\lambda$ moves from $0$ to $1$ are more intuitive. 
A disadvantage is that $D_\lambda$ itself is not Euclidean, leading to (typically only slightly) poorer-quality visualization in MDS plots (Section~\ref{sec:visualisation}).

\subsection{The role of $\lambda$}

\begin{figure}[htb]
\centering
\includegraphics[width=0.95\textwidth]{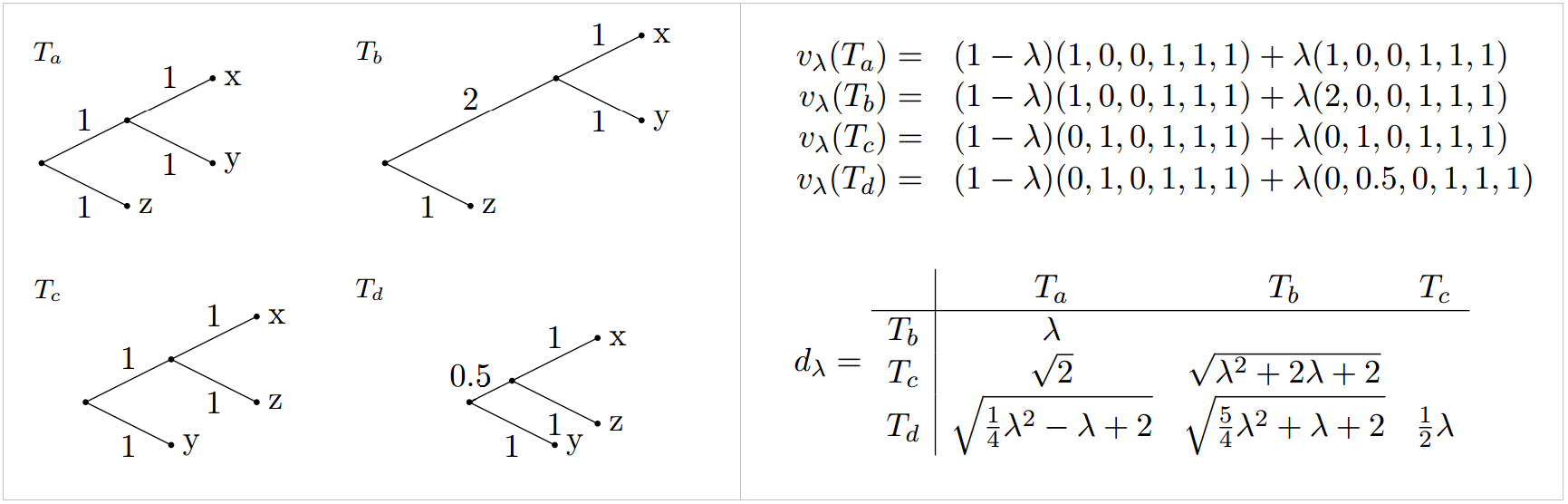}
\caption{Example trees from $\mathcal{T}_3$ to illustrate the effect of changing $\lambda$. 
The distance between $T_a$ and $T_c$ ($d_\lambda(T_a,T_c)$) is fixed for $\lambda \in [0,1]$ because their unmatched edges have the same length. $d_\lambda(T_b,T_d)<d_\lambda(T_b,T_c)$ for $\lambda \in (0,1]$ because the edge which $T_c$ and $T_d$ share and which is not found in $T_b$ is shorter in $T_d$ than in $T_c$.
Most entries increase with $\lambda$. 
The only distance to decrease as $\lambda \rightarrow 1$ is $d_\lambda(T_a,T_d)$, because the difference between the lengths of their unmatched branches is less than one.}
\label{fig:lambdaexample}
\end{figure}

The parameter $\lambda$ allows the user to choose to what extent the branch lengths of a tree, vs its topology alone, contribute to the tree distance.  
The distance between two trees may increase or decrease as $\lambda$ increases from $0$ to $1$. 
Since the topology-based vector, $m$, contains the number of edges along paths in the tree, and $M$ contains the path lengths, the branch lengths are implicitly compared to $1$ in the convex combination $v$. 
In other words, if the branch lengths are much larger than 1, then the entries of $M$ will be much larger than the corresponding entries of $m$, and $M$ will dominate in the expression for $v$ even when $\lambda$ is relatively small. 
Conversely, if the branch lengths are much less than 1, the entries of $M$ will be much less than those of $m$, and a value of $\lambda$ near 1 will be required in order for lengths to substantially change $v$.  In the case when all branch lengths are equal to 1, $m=M$ and the distance is independent of $\lambda$. 
The example in Figure~\ref{fig:lambdaexample} may provide some intuition. 

In order to capture length-sensitive distances between trees, we may wish to use a value of $\lambda$ such that neither $(1-\lambda)m$ nor $\lambda M$ dominate excessively, but naturally this will depend on the analysis.
For a more gradual change in $d_\lambda$ as $\lambda$ tends to 1, and for comparison of this change across different data sets, it is possible to rescale the branch lengths, for example by dividing all branch lengths by the median, or by changing the units.
However, this should be done with caution because information is inevitably lost through rescaling.
For example, if a phylogenetic analysis of multiple genes from the same organism had produced trees with similar topologies but different clock rates (e.g.\ branches in trees from gene 1 were typically twice as long as branches in trees from gene 2), this information would be obscured by rescaling.

\subsection{Other metrics on labeled phylogenetic trees} \label{sec:others}

Various metrics have been defined on phylogenetic trees.
For a recent comparative survey, see~\cite{Kuhner2014}.

The vector $M(T)$ is similar to the cophenetic vector of Cardona et al.~\cite{Cardona2013}, following Sokal and Rohlf~\cite{Sokal1962}, where $M_{i,j}$ is called the \emph{cophenetic value} of tips $i$ and $j$.
Parts (1) and (2) of our proof follow directly from results in~\cite{Cardona2013}.
Instead of the pendant branch lengths $p_i$, Cardona et al.\ use the depth of each taxon, which can be considered as $M_{i,i}$.
This involves a repetition of information between $M_{i,i}$, $M_{j,j}$ and $M_{i,j}$ whenever $M_{i,j}>0$.
However, their definition does allow for the presence of nested taxa (taxa which are internal nodes of the tree).
Cardona et al.\ also note that tree vectors such as these can be compared by any norm $L^p$, but that the Euclidean norm $L^2$, which we also use, has the benefits of being more discriminative than larger values of $p$, and enabling many geometrical and clustering methods.

The most widely used metric is that of Robinson-Foulds (RF)~\cite{Robinson1981}.
However, RF and its branch-length weighted version~\cite{Robinson1979} are fundamentally very different from our metric because they are defined on unrooted trees, whereas our metric emphasizes the placement of the root and all the descendant MRCAs.
Similarly, the path difference metrics of Williams and Clifford~\cite{Williams1971} and Steel and Penny~\cite{Steel1993} are for unrooted trees.
They compare the distance between each pair of tips in a tree; in essence, they consider the distance between \emph{tips} and their MRCA, whereas our metric considers the distance between the \emph{root} and the MRCA.
These metrics therefore capture different characteristics of trees and are only loosely correlated with our metric.

The metric introduced by Billera, Holmes and Vogtmann (BHV) captures branch lengths as well as tree structure~\cite{Billera2001} on rooted trees. 
The BHV tree space is formed by mathematically `gluing' together orthants. 
Each orthant corresponds to a tree topology and moving within an orthant corresponds to changing the tree's branch lengths. 
Moving from one orthant to an adjacent one corresponds to a nearest-neighbor interchange move. 
The metric is convex: for any two distinct trees $T_1$ and $T_2$, there is a tree $T_3$ `in between' them, i.e.\ such that $d^{BHV}(T_1,T_3)+ d^{BHV}(T_3,T_2) = d^{BHV}(T_1,T_2)$. 
This is a mathematically appealing and useful property, in part because it allows averaging of trees~\cite{Bacak2014}. However, it does not allow the user to choose a balance between the topology of the tree and the branch lengths. 
We provide further comparisons in Figure~\ref{fig:sixshapesupp}. 

Our metric compares trees with the same set of taxa (i.e.\ the same tips). As a consequence, it is suited for studies in which there is one set of taxa, and trees can be compared from different genes, inference methods, and sources of data.
Our metric does not capture distances between trees with different taxa; where the taxa overlap between two trees, our approach can compare the subtrees restricted to the taxa present in both trees. In contrast, comparisons between \emph{unlabeled} trees take a different form (e.g.\ kernel methods~\cite{Poon2013}), suitable to comparing trees on different sets of taxa.

Many phylogenetic analyses are, implicitly or explicitly, conducted in the context of a rooted tree. In the context of macroevolution, examples include estimates of times to divergence, ancestral relationships and ancestral character reconstruction. In more recent literature, most methods to link pathogen phylogenies to epidemic dynamics (phylodynamics) \cite{Stadler2011,Rasmussen2014,Didelot2014} are based on rooted phylogenetic trees.  For these reasons, the fact that the relationships to the root of the tree play a central role in our metric allows it to capture intuitive similarities in groups of trees in a way that other metrics do not.

\section{Exploring tree space} \label{sec:exploring}

Tree spaces are large and complex. 
It is important to understand the `shape' of a tree space before attempting to summarize it.
Our metric creates a space which can be effectively visualized (Section~\ref{sec:visualisation}) and where \emph{islands} (distinct clusters) of tree topologies can be detected. We demonstrate these techniques on a sample dataset of BEAST posterior trees for Dengue fever.
Finally, in Section~\ref{sec:summary} we describe how our metric can be used to make a principled selection of summary trees.

\begin{figure}[htb]
\centering
\subfloat[$\lambda=0$]{
\includegraphics[height=7.25cm]{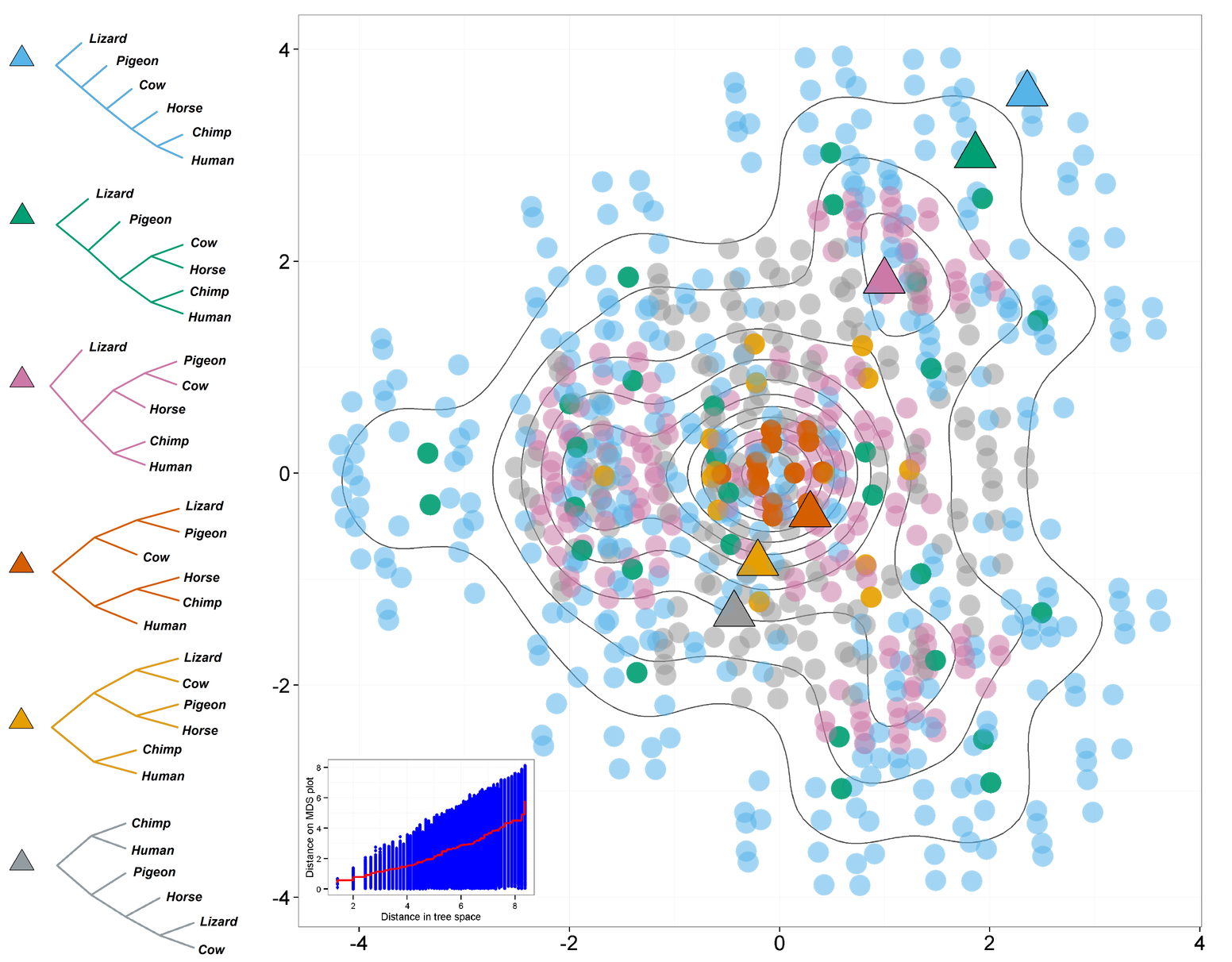}
\label{fig:sixshapesuppOurs}} \\
\subfloat[RF]{
\includegraphics[width=0.35\textwidth]{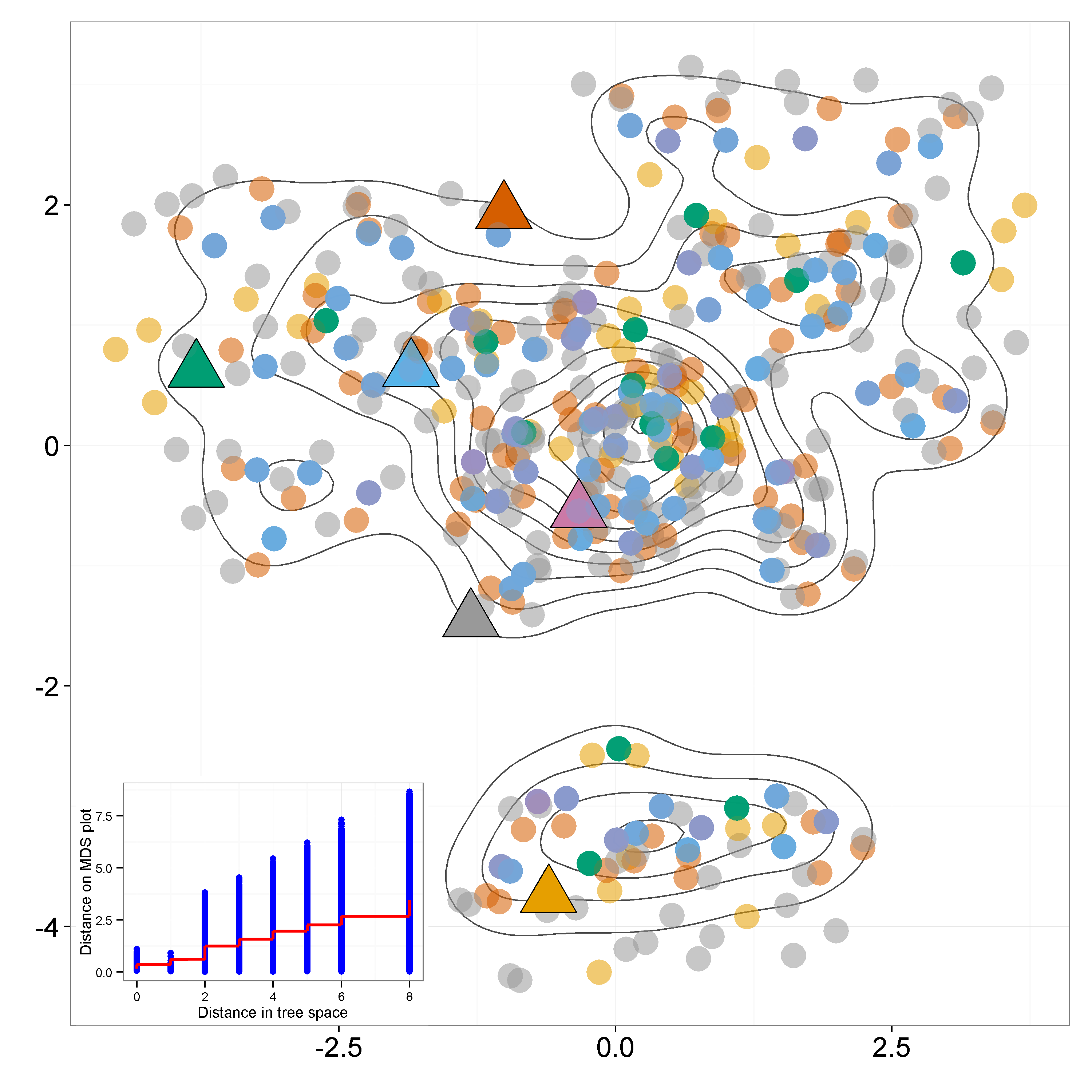}
\label{fig:sixshapesuppRF}}
\hspace{20pt}
\subfloat[BHV]{
\includegraphics[width=0.35\textwidth]{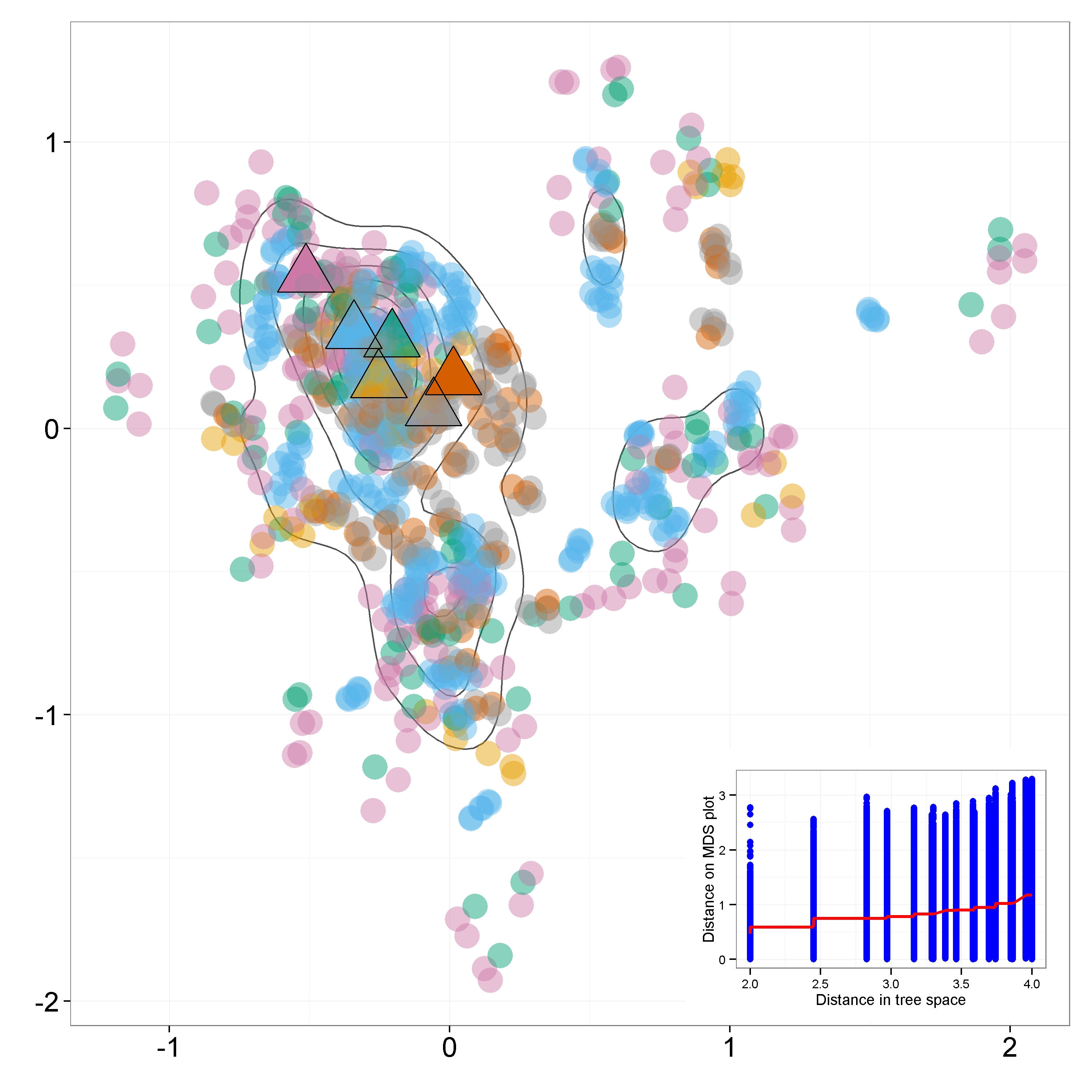}
\label{fig:sixshapesuppBHV}}
\caption{MDS projections of the shape of $\mathcal{T}_6$ according to metrics as shown, with corresponding Shepard plots. Colors correspond to tree shapes, of which examples are shown with triangles.
Symmetries correspond to permutations of the labels. In order to include the BHV metric in this comparison we assigned all branch lengths to be 1, with the result that $m=M$ and our metric is invariant to $\lambda \in [0,1]$.}
\label{fig:sixshapesupp}
\end{figure}

\subsection{Visualizing tree space} \label{sec:visualisation}

Visualization techniques like multidimensional scaling (MDS)~\cite{Cox2000} have been used to explore tree space previously, but are challenged by poor-quality projections~\cite{Hillis2005,Berglund2011}. 
When a set of distances is projected into a low-dimensional picture, there is typically some loss of information, which may result in a poor-quality visualization. For example, if 10 points are all 3 units away from each other, this will not project well into two dimensions; some will appear more closely grouped than others. 
However, if there are only 3 such points they can be arranged on a triangle, capturing the distances in two dimensions.

One approach to checking the quality of a visualization is a Shepard plot~\cite{Shepard1972}, which is a scatter plot of the true distance vs the MDS distance (i.e.\ the distance in the projection). 
Figure~\ref{fig:sixshapesupp} shows the MDS plot of the space of trees on 6 tips (with unit branch lengths) under our metric and two others: RF~\cite{Robinson1981} and  BHV~\cite{Billera2001}.
Shepard plots are included as an indication of the quality of each projection.

Each metric captures differences in both shape (shown by color) and labeling.  
Our approach produces a wide range of tree distances and captures intuitive similarities (e.g.\ the similar chimp-human pairing in the yellow and gray triangles in Figure~\ref{fig:sixshapesuppOurs}).
All 945 possible tree shapes and permutations of their labels are present in the input set of trees, and consequently there is no asymmetry that should lead to one group being separated from the rest.  Our metric captures the symmetry in the space and illustrates this in the MDS projection (Figure \ref{fig:sixshapesuppOurs}), whereas in RF and BHV (Figures~\ref{fig:sixshapesuppRF} and~\ref{fig:sixshapesuppBHV}), poor-quality projections lead to apparent distinct tree islands where none exist.   
This makes detecting genuine islands in posterior sets of trees difficult using RF or BHV. 
The Euclidean nature of our metric means that it is well-suited to visualizations that project distances into two- or three-dimensional Euclidean space. 
The Shepard plots illustrate that the correspondence between the projected distances and true distances is better in our metric than the others, though the projection distance can be much smaller than the true distance (but not the converse). 
MDS projections are of higher quality for trees from data than in the space of all trees on 6 tips (e.g.\ Figure~\ref{fig:dengue}).

\subsection{Islands in tree space} \label{sec:islands}

Tree inference methods explore the set of possible trees given the data, but there are many alternative trees. 
Bayesian Markov Chain Monte Carlo (MCMC) methods as implemented in BEAST~\cite{Drummond2012} and MrBayes~\cite{MrBayes2001} produce a posterior set of trees, each with associated likelihoods.  
Distinct islands of trees within small NNI distance can share a high parsimony or likelihood~\cite{Maddison1991,Salter2001}. Complicating matters further, not all taxa in a dataset will have complete data at all loci. In this case, there are `terraces' of many equally likely trees, with trees in a terrace all supporting the same subtrees for the taxa with data at a given locus~\cite{Sanderson2011}. These facts have deep implications for tree inference and analysis, but the difficulty of detecting and interpreting tree islands has meant that the majority of analyses, particularly on large datasets, remain based on a single summary tree method such as the maximum clade credibility (MCC) tree with posterior support values illustrating uncertainty, or on maximum likelihood or parsimony trees with bootstrap supports. 
Our metric can detect distinct clusters or islands of close tree topologies ($\lambda=0$) within a collection of trees. Since distance is defined by the metric that is used, these are different from previously described tree islands~\cite{Maddison1991,Salter2001}.

\begin{figure}[htb]
\centering
\subfloat[GI relaxed clock, $\lambda=0$]{
\includegraphics[width=9cm]{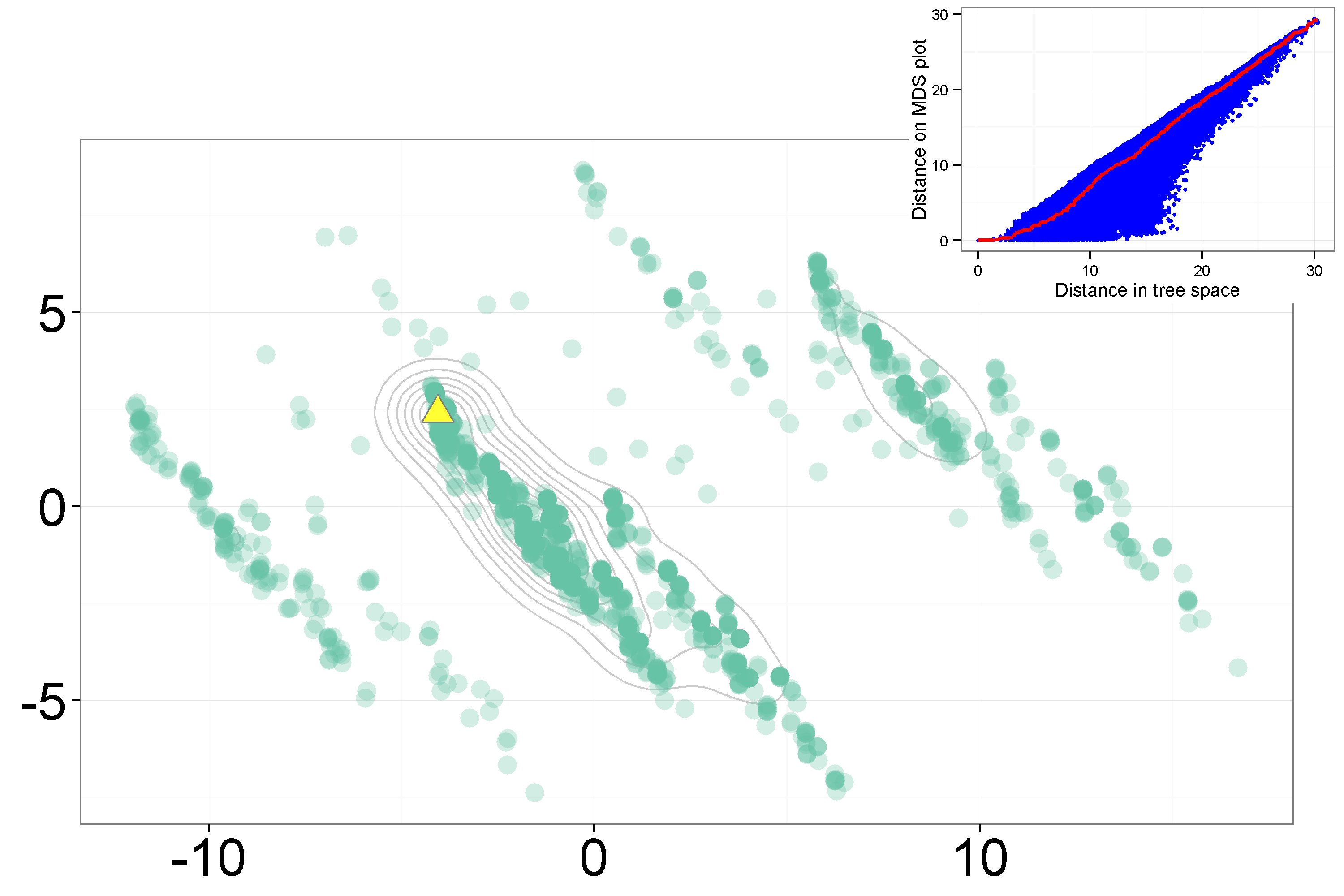}
\label{fig:dengueGI}} \\
\subfloat[CP strict clock, $\lambda=0$]{
\includegraphics[width=9cm]{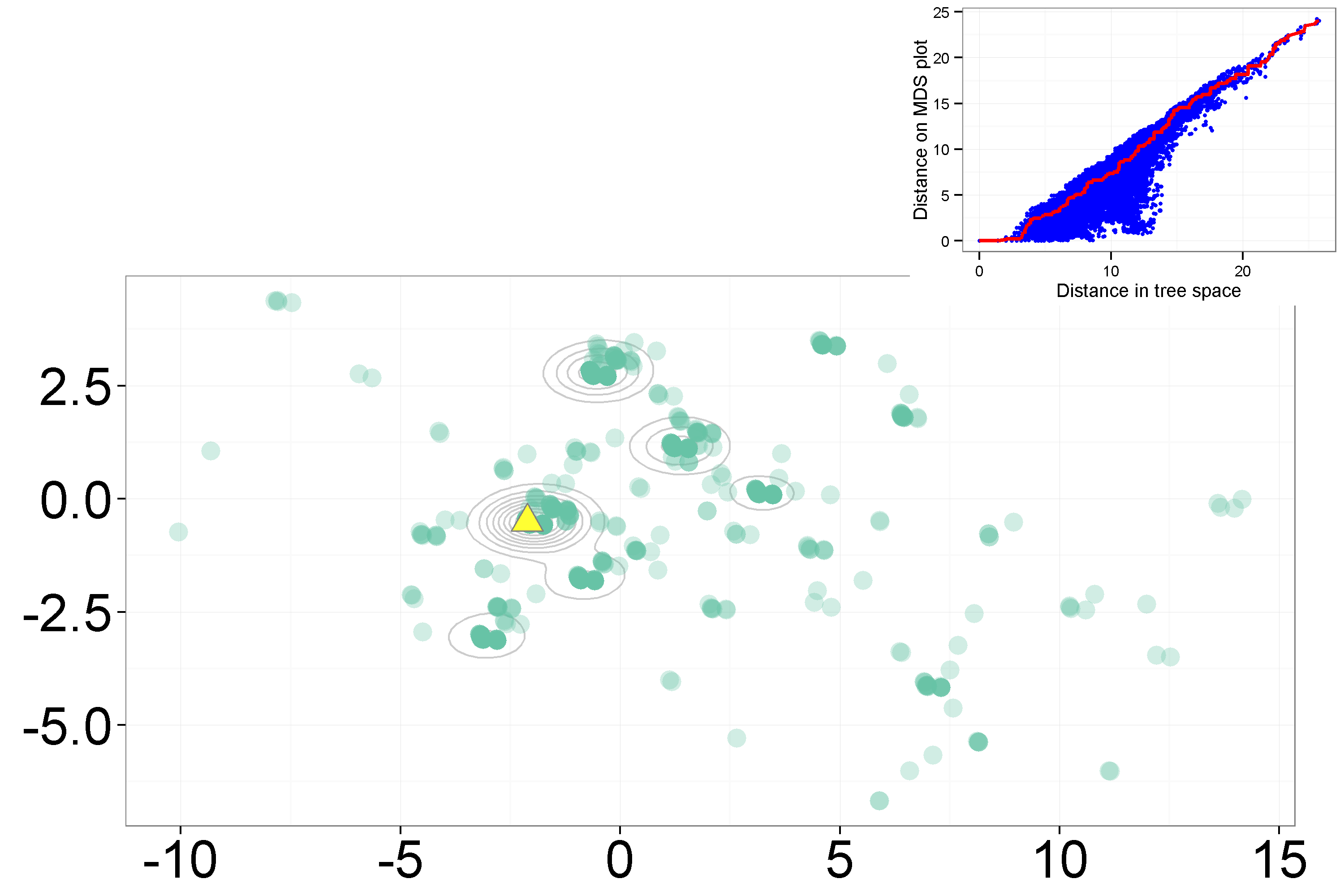}
\label{fig:dengueCP}}
\caption{MDS plots of dengue fever trees sampled from posteriors demonstrate differences in the space of trees explored by BEAST under different settings. MCC trees are marked by yellow triangles. (a) GTR + $\Gamma$ + I substitution model with uncorrelated lognormal-distributed relaxed molecular clock (b) Codon-position specific substitution model GTR~+~CP, with a strict clock.
}
\label{fig:dengue}
\end{figure}

We demonstrate our approach using the examples from
the original paper introducing BEAST~\cite{Drummond2007}, where Drummond and Rambaut demonstrated their Bayesian analysis on 17 dengue virus serotype 4 sequences from~\cite{Lanciotti1997} under varying priors for model and clock rate.
As a means of comparing posterior tree distributions under different BEAST settings, we ran the \texttt{xml} files provided in~\cite{Drummond2007} through BEAST v1.8 and analyzed the resulting trees.
In Figure \ref{fig:dengue} we demonstrate MDS plots of two of these analyses: Figure \ref{fig:dengueGI} is a sample of the posterior under the standard GTR + $\Gamma$ + I substitution model with uncorrelated lognormal-distributed relaxed molecular clock; Figure \ref{fig:dengueCP} is a sample from the posterior under the codon-position specific substitution model GTR + CP, with a strict clock.
These analyses demonstrate some of the different signals which can be detected by visualizing the metric's tree distances: distinct islands are visible in (a), whereas in (b) there are some tight bunches of points but the posterior is not as clearly separated into distinct islands.
Additionally,  trees in (b) are more tightly grouped together, indicating that is less conflict in the phylogenetic signals in (b).
We ran BEAST twice with the settings from (a) (using different random starting seeds), and found that the space of trees explored and accepted in each run was  similar, with the same islands. It is also encouraging that the MCC tree from the first BEAST run had the same topology as that from the second run, and that this topology sits in the largest island (yellow triangle in Figure \ref{fig:dengueGI}).
Similarly, the MCC tree is in the largest cluster in (b).

Islands are of concern for tree inference and for outcomes that require the topology of tree, which will affect ancestral character reconstruction and consequently the interpretation of many phylogenetic datasets~\cite{Sullivan1996}.
However, other analyses, and tree estimation methods themselves, take trees' branch lengths as well as topology into account. We find that islands typically merge together in the metric as $\lambda$ approaches $1$; the posterior becomes unimodal.

\subsection{Summary trees} \label{sec:summary}

Summarizing groups of phylogenetic trees is challenging, particularly when there are different alternative and inconsistent topologies~\cite{Heled2013}. MCC trees can summarize posterior distributions; they rely on including the clades with the strongest posterior support but where these are not concordant the resulting MCC trees can have negative branch lengths. Furthermore, the MCC tree itself may never have been sampled by the MCMC chain, casting doubt on its ability to reflect the relationships in the data. 

Our metric allows us to find `central' trees within any group of trees: a posterior set of trees, or any island or cluster of trees.
To do this, we exploit the fact that our metric is simply the Euclidean distance between the two vectors $v_\lambda(T_a)$ and $v_\lambda(T_b)$.
Among $N$ trees $T_i$ $(i=1,\dots,N)$ in a posterior sample, we can find the tree closest to the average vector $\bar{v} = \frac{1}{N}\sum_{i=1}^N v_\lambda(T_i)$. The average vector $\bar{v}$ may not in itself represent a tree, but we can then find the tree vectors from our sample which are closest to this average.
These vectors correspond to trees, $T_c$, (not necessarily unique) which minimize the distance between $\bar{v}$ and $v_\lambda(T_c)$. This minimal distance is a measure of the quality of the summary: if it is small,  $T_c$ is close to `average' in the posterior. 
$T_c$ is known as the geometric median tree~\cite{Haldane1948}. 
The geometric median is one of a range of barycentric methods which can be used with our metric to select a tree as a representative of a group. It is also straightforward to weight trees by likelihood or other characteristics when finding the geometric median. This provides a suite of tools for summarizing collections of trees.  Geometric median trees will always have been sampled by the MCMC, and will not have negative branch lengths. We found that within islands, geometric median trees are very close to the MCC tree for the island.

\section{Concluding remarks} \label{sec:conclusion}

The fact that our metric is a Euclidean distance between two vectors whose components have an intuitive description means that simple extensions are straightforward to imagine and to compute. For example, it may be the case that the placement of a particular tip is a key question. This could occur, for example, in a real-time analysis of an outbreak, where new cases need to be placed on an existing phylogeny to determine the likely source of infection. We could form a metric that emphasizes differences in the placement of a particular tip (say, $A$), by weighting $A$'s entries of $m$ and $M$ highly compared to all other entries. In this new metric, trees would appear similar if their placement of $A$ was similar; patterns of ancestry among the other tips would contribute less to the distance. 
Indeed, it is possible to design numerous metrics, extending this one and others, and using linear combinations of existing metrics~\cite{Liebscher2015}. 

Our metric enables quantitative comparison of trees.
It is relevant to viral, bacterial and higher organisms and can help to reveal distinct, likely patterns of evolution. 
It allows quantitative comparison of tree estimation methods and can provide a heuristic for convergence of tree estimates.  
There are also many applications in comparing trees derived from different data. For example, the metric can be used to detect informative sites which, when removed from sequence alignments, change the phylogeny substantially.
More generally, our metric can find distances between any rooted, labeled trees with the same set of tips.
It can be used to compare tree structures 
from a variety of scientific disciplines,
including decision trees, network spanning trees, hierarchical clustering trees and language trees. 

%\bibliography{library}
%\bibliographystyle{plain}

\end{document}